\documentclass[a4paper,11pt]{article}
\usepackage[margin=0.95in]{geometry}

\usepackage{multirow}
\usepackage{xspace}
\usepackage{paralist}
\usepackage[utf8]{inputenc}
\usepackage[T1]{fontenc}
\usepackage{amsmath} 
\usepackage{amsthm}
\usepackage{amssymb}
\usepackage{graphicx}
\usepackage{wrapfig}
\usepackage{complexity}
\usepackage{cite}
\usepackage{subfigure}
\usepackage{booktabs}
\usepackage{mathtools}

\newcommand{\myparagraph}[1]{\vspace*{.8ex}\noindent #1}

\usepackage[vlined,linesnumbered,ruled]{algorithm2e}
\DontPrintSemicolon

\newtheorem{theorem}{Theorem}

\usepackage{booktabs}

\usepackage{todonotes}
\usepackage{hyperref}
\usepackage{url}

\usepackage{placeins}

\usepackage[pagewise]{lineno}

\DeclareMathOperator{\length}{len}

\DeclareMathOperator{\lFont}{fs}

\DeclareMathOperator{\lName}{rn}
\DeclareMathOperator{\lStroke}{st}
\DeclareMathOperator{\lColor}{co}
\DeclareMathOperator{\covered}{cov}

\newcommand{\MaxTotalCovering}
{\textsc{Max\-Labeled\-Roads}\xspace}

\newcommand{\ILong}{\textsc{Issue~3}\xspace}
\newcommand{\ILanes}{\textsc{Issue~1}\xspace}
\newcommand{\ICrossings}{\textsc{Issue~2}\xspace}
\newcommand{\IFatLabels}{\textsc{Issue~4}\xspace}

\newcommand{\RuleA}{\textit{Rule~1}\xspace}
\newcommand{\RuleB}{\textit{Rule~2}\xspace}
\newcommand{\RuleC}{\textit{Rule~3}\xspace}
\newcommand{\RuleD}{\textit{Rule~4}\xspace}

\newcommand{\GreedyAlgo}{\textsc{Base\-Line}\xspace}
\newcommand{\TreeAlgo}{\textsc{Tree}\xspace}
\newcommand{\ILPAlgo}{\textsc{Milp}\xspace}
\newcommand{\Shredder}{\textsc{D\&C}}

\title{An Algorithmic Framework for Labeling Road Maps}

\author{Benjamin Niedermann\thanks{Institute of Theoretical Informatics, Karlsruhe Institute of Technology, Germany} 
  \and Martin~N\"ollenburg\thanks{Algorithms and Complexity Group, TU Wien, Vienna, Austria} 
} 
\date{}

\begin{document}
\maketitle

\begin{abstract} 
  Given an unlabeled road map, we consider, from an algorithmic
  perspective, the cartographic problem to place non-overlapping road
  labels embedded in their roads.  We first decompose the road
  network into logically coherent road sections, e.g., parts of
  roads between two junctions.  Based on this decomposition, we present and implement a new and
  versatile framework for placing labels in road maps such that the
  number of labeled road sections is maximized. 
  In an experimental evaluation with road maps of 11 major cities we show that our proposed labeling algorithm is both fast in practice and that it reaches near-optimal solution quality, where optimal solutions are obtained by mixed-integer linear programming. In comparison to the standard OpenStreetMap renderer Mapnik, our algorithm labels~$31\%$ more road sections in average.
  %
\end{abstract}

\section{Introduction}
Due to the increasing amount of geographic data and its steady change,
automatic approaches become more and more important in the area of
cartography. This particularly applies to the time-consuming and demanding task of label
placement and much research has been done on its automation.
Badly placed labels of feature of interest make maps easily
unreadable~\cite{imhof}. Depending on the type of map feature the label
placement is done differently. For \emph{point features} (e.g., cities
on small-scale maps) labels are typically placed closely to that
feature, while for \emph{line features} (e.g., roads, rivers) the name
is either placed along or inside the feature. The latter approach is
also used for \emph{area features} (e.g., lakes). Independently of the applied technique and feature type, labels should not overlap each other and
clearly identify the features~\cite{criteria}. 

The cartographic label placement problem has also attracted the interest of researchers in computational geometry and it has been thoroughly investigated from both the practical and theoretical perspective~\cite[Chapter
58.3.1]{overview},~\cite{bibliography}.
While algorithms for labeling point features
got a lot of attention, much less work has been done on line
features and area features. In this paper we address
labeling line features, namely labeling the entire road network of a road map.  We take an algorithmic, mathematical perspective on the underlying optimization
problem and build up on our recent theoretical results for labeling tree-shaped networks~\cite{rlTheory}.  We apply the quality
criteria for label placement in road maps elaborated by
Chirié~\cite{street-name-placement} based on interviews with
cartographers. They include that (C1) labels are placed inside and
parallel to the road shapes, (C2) every road section between two
junctions should be clearly identified, and (C3) no two road labels
may intersect. Similar criteria have been described in a classical paper by~Imhof~\cite{imhof}. 

 Variations of embedded labels have been considered in road maps
 before. Chirié~\cite{street-name-placement} and Strijk et
 al.~\cite[Ch. 9]{strijk2001} presented simple, local heuristics that
 place non-overlapping labels based on a discrete set of candidate
 positions -- in contrast we consider the problem globally applying a
 continuous sliding model.  Seibert and Unger~\cite{labelingManhattan}
 utilized the geometric properties of grid-based road networks and
 proved that it is NP-complete to decide whether at least one label
 can be placed for each road. For the same grid-based setting Neyer
 and Wagner~\cite{downtownLabeling} evaluated a practically efficient
 algorithm that is not applicable for general road networks.

Road labeling with embedded labels has also been considered for
interactive and dynamic maps. Maass and Döllner~\cite{Maass07} provided a
heuristic for labeling interactive 3D road maps taking
obstacles into account. Vaaraniemi et al.~\cite{Vaaraniemi12} presented a
study on a force-based labeling algorithm for dynamic maps considering
both point and line features. Schwartges et al.~\cite{swh-lsime-acm14}
investigated embedded labels in interactive maps allowing panning,
zooming and rotation of the map. They evaluated a simple
heuristic for maximizing the number of placed labels. 

In contrast, non-embedded labels are typically considered for single
line features such as rivers. Edmondson et
al.~\cite{Edmondson96} presented an algorithm for placing straight
labels along single line features. Wolff et al.~\cite{wkksa-seahq-00}
also considered the case that labels may bend. Recently, Schwartges et
al.~\cite{Schwartges2015} used \emph{billboards} (labels with short leaders) for naming roads in interactive 3D maps to avoid label~distortion.

\begin{figure}[t]
\centering
\includegraphics[scale=7.0]{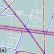}
\includegraphics[scale=7.0]{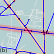}
\includegraphics[scale=7.0]{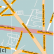}
\includegraphics[scale=7.0]{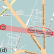}
\caption{The presented workflow. (a) The road network given by
  polylines (blue segments). (b) Phase 1:
  A graph~$G$ is created whose embedding is the simplified road
  network; blue segments: road sections, red segments: junction
  edges. (c) Phase 2: Creating the
  labeling using~$G$. (d)~A labeling produced
  by the OSM renderer Mapnik. The six labels of road \emph{Osloer Straße} are enclosed by red ellipses. }
\label{fig:motivation}
\end{figure}

\begin{figure}[t]
\centering
\subfigure[Google Maps:  Hermann-Vollmerstraße, Karlsruhe (Germany)]{\includegraphics[page=2,scale=1.0]{./fig/example.pdf}\label{example:google}}
\subfigure[Bing Maps: Kirchbuel, Karlsruhe (Germany)]{\includegraphics[scale=1.0]{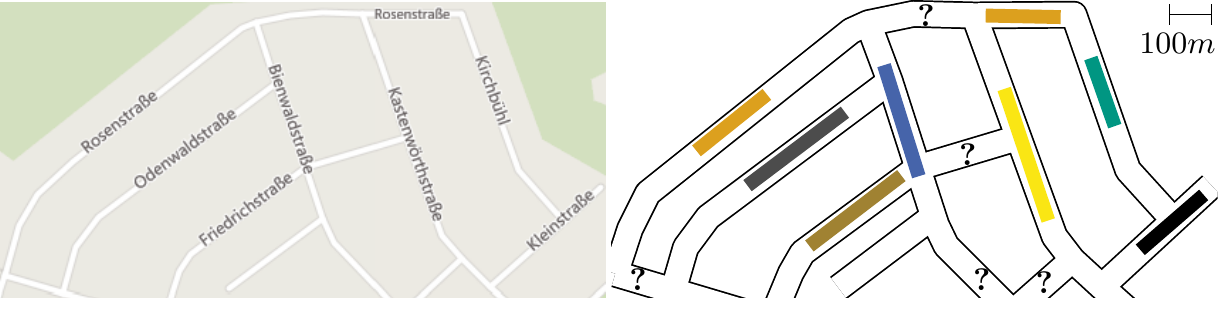}\label{example:bing}}
\caption{Two maps of the map services Google Maps and Bing Maps. Left: Screen shot of a road map. Right: The redrawn road map to emphasize the labels. Labels of the same color belong to the same road. \subref{example:google} Labels are placed tightly packed in a row. While some roads have more labels then necessary, other roads are not labeled. \subref{example:bing} Road map consists of many unlabeled road~sections.}
\label{fig:webservices}
\end{figure}

For labeling point features a typical objective is to maximize the
number of non-overlapping placed labels, because every placed label
enhances the map with further information. While this is mostly
true for point features, maximizing the number of labels is not the
right objective for label placement of roads since not every label
that is placed necessarily contributes more information to the
map. For example, consider the placed labels of the road \emph{Osloer Straße}
in Fig.~\ref{fig:motivation}(d).  We can
easily remove some of those labels without losing any information,
because the map user can still identify the same road sections; see
Fig.~\ref{fig:motivation}(c). In online map services,
however, one often can find such redundant labels; see
Fig.~\ref{fig:webservices} for two examples.  Some
roads may have unnecessarily many labels, which may in turn cause others to remain completely
unlabeled. Hence, the user
cannot identify such roads on the map, a real disadvantage if headed
for that road.  Due to these observations we do not aim to maximize
the number of labels, but the number of labeled \emph{road
  sections}. For the purpose of this paper, a \emph{road section}
forms a connected piece of the road network that logically belongs together,
e.g., a part of a road between two junctions or a part that stands out
by its color or width. Our algorithm, however, is independent of the actual definition of road sections; any partition of the road network into disjoint road sections can be handled. We say that a road section is \emph{labeled}
if a label (partly) covers it.

As the underlying model for maximizing labeled road sections we re-use the planar graph model that has been introduced in our theoretical companion paper~\cite{rlTheory}. In that paper we further proved that labeling a maximum number of road sections is
NP-hard, even for planar graphs and if no road consists of multiple branches. However, we presented a polynomial-time algorithm for the case that the road graph
is a tree.  While the result for trees is mostly of theoretic
interest (road networks rarely form trees), we will show in this paper
that our tree-based algorithm can in fact be used successfully as the
core of an efficient and practical road labeling algorithm that
produces near-optimal solutions.


\paragraph{Contribution \& Outline.} We introduce a
new, versatile algorithmic framework for placing non-overlapping
labels in road networks maximizing the number of labeled road sections.  We
keep the algorithmic components easily exchangeable.  In
Sect.~\ref{sec:model} we discuss and expand the model introduced
in~\cite{rlTheory}. Afterwards, we present a workflow for labeling road networks consisting of two phases; see Fig.~\ref{fig:motivation}.

\textit{Phase 1 (Sect.~\ref{sec:transformation}).} We translate the
given road network into a semantic representation (an abstract road
graph) that identifies pieces of the road network that belong
semantically together. To that end, we simplify the road network,
e.g., we merge lanes closely running in parallel. By
  design this simplification maintains the overall geometry of the road
  network and only merges structures in the data that should not be labeled independently. Phase~1 is not part of the labeling optimization process.

\textit{Phase 2 (Sect.~\ref{sec:labeling}).} Based on
the abstract road graph, we create an actual labeling using one of three  presented
algorithms: a naive base-line algorithm, a heuristic
extending our tree-based algorithm~\cite{rlTheory} and a
 mixed-integer linear programming (MILP) formulation. 


 As proof of concept we implemented the core of the framework only taking 
    the most important cartographic criteria into account. However, with
   some engineering it can be easily enhanced to more complex models,
   e.g., enforcing minimum distances between labels, abbreviating road names, or using alternative definitions of road sections.  In
 Sect.~\ref{sec:evaluation} we present a detailed evaluation of our
 framework on 11 sample city maps.  Due to its availability and
 popularity in practice, we compare our results against the standard
 OpenStreetMap (OSM) renderer Mapnik as a representative of local
 heuristics; it uses a strategy similar
 to~\cite{street-name-placement,strijk2001}. We show that our
 tree-based algorithm is fast and yields near-optimal labelings that
 improve upon Mapnik by $31\%$ on average.


\begin{figure}[htb]
\centering
\includegraphics[scale=1.2]{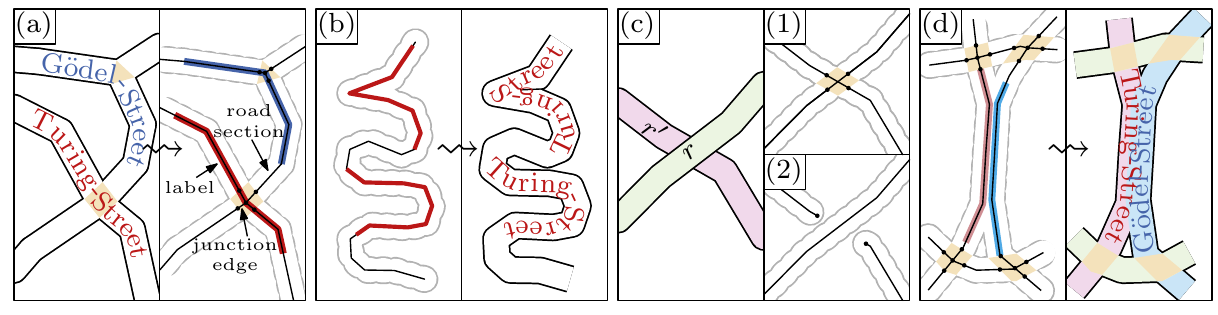}
\caption{Illustration of model and arising issues. (a) Sketch of a road network and its abstract road graph.  (b) Labels are possibly curvy and have sharp bends making the text hardly legible. (c) \ICrossings: Two ways to represent bridges and tunnels in the abstract road graph. (d) \IFatLabels: The text representation of labels  may overlap, while the curve representation in the abstract road graph does not.}
\label{fig:model-issues}
\end{figure}

\section{Semantic Representation of Road Networks}\label{sec:model}


At any given scale, road networks are typically drawn as follows. Each
road or road lane is represented as a thick, polygonal curve, i.e.,
a polygonal curve with non-zero width; see the background of Fig.~\ref{fig:motivation}(a). If two (or more) such curves intersect, they form
junctions. If two or more lanes of the same road closely run in parallel they merge to one even
thicker curve such that individual lanes become indistinguishable.  We
then want to place road labels inside those thick curves. More
precisely, a \emph{road label} can again be represented as a thick
curve (the bounding shape of the road name) that is contained in and
parallel to the thick curve representing its road; see~Fig.~\ref{fig:motivation}(c).

For the purpose of this paper it is sufficient to use a simplified representation, which represents the road network
and its labels as thin curves instead~\cite{rlTheory}. More precisely, a road
network is modeled as a planar embedded \emph{abstract road graph} whose edges
correspond to the skeleton of the actual thick curves.  In this model a
label is again a thin curve of certain length that is contained in the
skeleton. Following the cartographic quality criteria~(C1)--(C3),
we want to place labels, i.e., find sub-curves of the skeleton,
such that
\begin{inparaenum}[(1)]
\item each label starts and ends on road sections, but not on junctions,
\item no two labels overlap, and
\item a maximum number of road sections are labeled. 
\end{inparaenum}
Requiring that labels end on road sections avoids ambiguous placement
of labels in junctions where it is otherwise unclear how the road passes through
it. Note that this does not forbid labels across junctions. From a
labeling of the abstract road graph it is straight-forward to
transform each label back into its \emph{text representation} by
placing the individual letters of each label along the thick curves; see
Fig.~\ref{fig:model-issues}(a).

\paragraph{Abstract Road Graph Model.}
We have introduced the abstract road graph in~\cite{rlTheory}, but for
the convenience of the reader we repeat it here, see also
Fig.~\ref{fig:motivation}(b) and Fig.~\ref{fig:model-issues}(a).   A road
network (in an abstract sense) is a planar geometric \emph{graph}
$G=(V,E)$, where each vertex $v\in V$ has a position in the plane and
each edge~$\{u,v\}\in E$ is represented by a polyline whose end points
are $u$ and $v$.  Each edge further has a \emph{road name}. A
maximal connected subgraph of $G$ consisting of edges with the same name
forms a \emph{road}~$R$. The length of the name of $R$ is denoted by~$\lambda(R)$.
Each edge $e \in E$ is either a \emph{road section}, i.e., the part of a road in between two junctions, or a \emph{junction edge}, which
models road junctions. Formally, a \emph{junction} is a maximal connected
subgraph of~$G$ that only consists of junction edges. We require that no two road sections in $G$ are incident to the same vertex and that vertices incident to road sections have at
most degree~2. Thus, the road graph $G$ decomposes into road sections, separated by junctions.

We say a point~$p$ lies on $G$, if there is an edge $e\in E$ whose
polyline contains~$p$. Hence, a polyline~$\ell$ (in particular a
single line segment) lies on~$G$ if each point of~$\ell$ lies
on~$G$. Further, $\ell$ \emph{covers} $e$, if there is a point
of~$\ell$ that lies on~$e$. If each point of $e$ is covered by~$\ell$,
$e$ is \emph{completely covered}.  The \emph{geodesic distance} of two points
on~$G$ is the length of the shortest polyline on $G$ connecting both
points.

A \emph{label} of a road~$R$ is a simple open polyline~$\ell$ on~$G$
that has length~$\length(\ell) = \lambda(R)$, ends on road sections
of~$G$, and whose segments only lie on edges of~$R$.  The start point
of $\ell$ is denoted as the \emph{head} $h(\ell)$ and the endpoint as
the \emph{tail} $t(\ell)$.  Obviously, the edges that are covered by
$\ell$ form a path~$\mathcal P_{\ell}=(e_1,e_2,\cdots,e_{k-1},e_k)$
such that $e_1$, and $e_k$ are (partly) covered and
$e_2,\dots,e_{k-1}$ are completely covered by~$\ell$. If $e_i$ is a road section (and not a junction
edge), we say that~$e_i$ is \emph{labeled}~by~$\ell$.

We extend the above abstract road graph model and restrict ourselves to \emph{well-shaped} labels, i.e., labels that are not too curvy or do not contain broken type setting due to sharp bends; see Fig.~\ref{fig:model-issues}(b). Similar to
  Schwartges et al.~\cite{swh-lsime-acm14}, we apply a local criterion
  to decide whether a label is well-shaped.
To that end, we define a label $\ell$ to be \emph{well-shaped} if for each covered
edge $e\in \mathcal P_\ell$ there is a \emph{well-shaped} piece of~$e$
that completely contains the part of $\ell$ on~$e$.  Further, we
require that for each pair of incident edges of $\mathcal P_\ell$ the
bend angle is at most $\alpha_\text{max}$, where $\alpha_\text{max}$
is a pre-defined constant. We redefine a \emph{labeling}~$\mathcal L$
to be a set of mutually non-overlapping, well-shaped labels.
Our theoretic results~\cite{rlTheory} remain valid for this
restriction. In particular only few minor technical adaptions are required for the tree labeling algorithm.

In order to identify \emph{well-shaped} pieces of a polyline~$P$ with
edges~$e_1,\dots,e_k$, we extend the approach presented
by Schwartges et al.~\cite{swh-lsime-acm14}. They define the curviness $w(P)$
of $P$ by summing up the bend angles $\alpha_i$ of all incident edge pairs
$e_i$, $e_{i+1}$, i.e., $w(P)=\sum_{i=1}^{k-1}|\alpha_i|$ to determine the best label positions for any given label. We want to locally classify road pieces as well-shaped instead and adapt their idea as follows. 
Let $S$ be a maximal sub-polyline of $P$ with the property that any sub-polyline of $S$ with length at most $l_\text{max}$ has curviness at most $\alpha_\text{max}$. Each such sub-polyline $S$ forms a well-shaped piece of $P$ and they can all be computed in $O(k)$ time. 
This local criterion for well-shapedness is based on the curviness of a fixed-width window sliding along the polyline; it is independent of the label length (similarly to what Mapnik does). 
In our experiments we set~$l_\text{max}$ to twice the length of the
letter \texttt{W} and $\alpha_\text{max}=22.5^\circ$, analogously to the parameters that Mapnik~uses.

A \emph{labeling} $\mathcal L$ for a road network is a set of mutually
non-overlapping, well-shaped labels, where 
 two labels $\ell$ and
$\ell'$ \emph{overlap} if they intersect in a point that is not their
respective head or tail. 
%
Following the criteria~(C1)--(C3), the problem
\MaxTotalCovering is to find a labeling $\mathcal L$ that labels a
maximum number of road sections, i.e., no other labeling 
labels more road sections.  In~\cite{rlTheory} we showed that
\MaxTotalCovering is NP-hard in general and can be solved in
$O(|V|^3)$ time if~$G$ is a tree.

\paragraph{Shortcomings for Real-world Road Networks.}\label{sec:issues} 
While the abstract road graph model allows theoretical insights, we
cannot directly apply it to real-world road networks. Due to the following
issues, we need to invest some effort in a preprocessing phase (see Sect.~\ref{sec:transformation}) to guarantee that the resulting
labels in the text representation do not overlap, look nicely and are embedded in the roads' shapes. 

\ILanes: If lanes run closely in parallel, their drawings in the road
network merge to one thick curve and individual lanes become
indistinguishable. Hence, in our abstract model, such lanes should be aggregated to a single road
section that represents the skeleton of the merged curve, 
and labels should be contained in it; see Fig~\ref{fig:motivation}(c).

\ICrossings: Real-world road networks are not planar, but edges may cross,
namely at tunnels and bridges; see Fig.~\ref{fig:model-issues}(c).
To avoid
overlaps between labels placed on those road sections, we either can model
the intersection as a regular junction of two roads or
we split one into two shorter road sections that do not cross the other road section. In both
cases the road graph becomes planar. For our
prototype we use the first variant (also used by Mapnik), 
because more road sections can be labeled.

\ILong: In real-world road networks some road sections are possibly so long
that the label should be repeated after appropriate distances.

\IFatLabels: Labels have a certain
font size so that when transforming an abstract label curve
into its text representation, labels of different roads may overlap due to their road sections being too close; see Fig.~\ref{fig:model-issues}(d).

\section{Phase 1 -- Construction of Abstract Road Graphs}\label{sec:transformation}

The first phase of our framework consists of transforming the input road network data into an abstract road graph while resolving the four issues mentioned in Sect.~\ref{sec:issues}.
Typically, road networks are
given as a set of polylines that describe the 
roads and road lanes. 
Individual polylines do not necessarily form semantic components such as road sections.
So as a first step, we break all polylines down into individual line segments (whose union forms the road network).
Let $L$ be the set of all these line segments.
We further require that each line segment $l\in L$
is annotated with its \emph{road name} $\lName(l)$, the stroke width
$\lStroke(l)$ and the color $\lColor(l)$ that are used to draw $l$,
and finally the \emph{font size} $\lFont(l)$ that shall be used to
display the name. 
We say that two line
segments $l,l'\in L$ are \emph{equally represented} if
$\lStroke(l)=\lStroke(l')$ and $\lColor(l)=\lColor(l')$. We assume
that $\lFont(l) < \lStroke(l)$ for any~$l$; otherwise we set
$\lStroke(l):=\lFont(l)$.

\begin{figure}[t]
\centering
\includegraphics[scale=1.2]{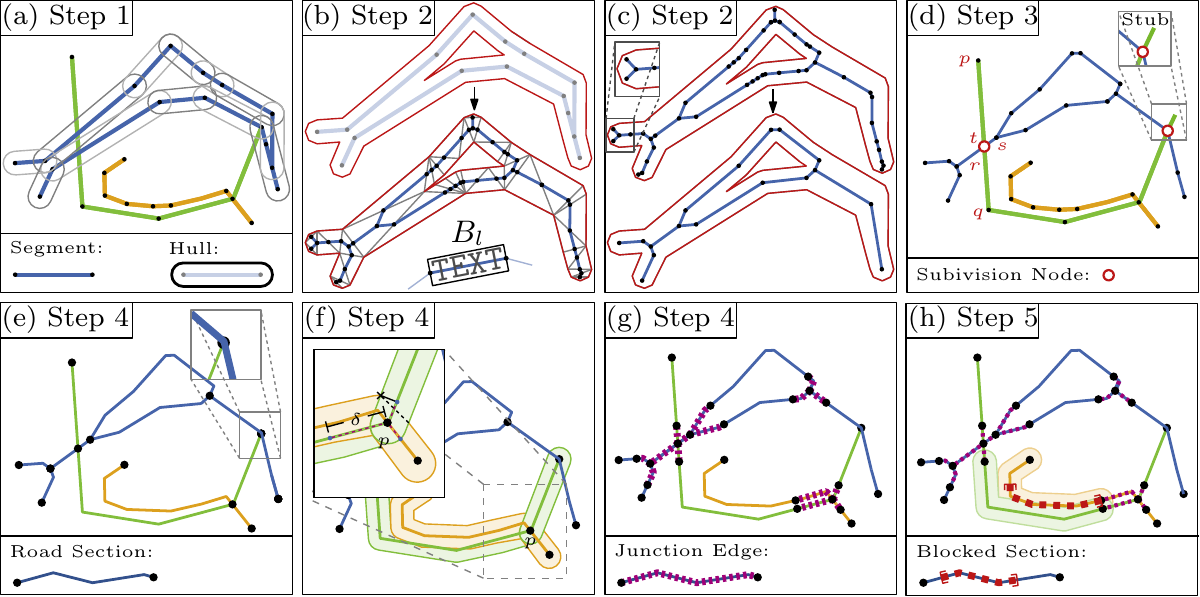}
\caption{Illustration of the steps applied in Phase 1. Segments of the same color have the same road name. For more details see the description of Phase 1. }
\label{fig:workflow}
\end{figure}

The workflow consists of the following five
steps; see Fig.~\ref{fig:workflow}.
\begin{inparaenum}[(1)]
\item \textsc{Identification.} Identify single \emph{road
    components}, i.e., sets of line segments that have the same name,
  are equally represented, and form a connected component.
\item \textsc{Simplification.} Simplify each road component such that lanes  running closely in parallel are aggregated. 
\item \textsc{Planarization.} Replace bridges and tunnels by artificial junctions. 
\item \textsc{Transformation.} Transform the segment representation into an abstract road graph. 
\item \textsc{Resolving Overlaps.} Identify mutual overlaps of road sections  and block them for label placement.
\end{inparaenum}

Below we describe each step in more detail. We define the
\emph{hull} of a line segment~$l\in L$ to be the region of points
whose Euclidean distance to~$l$ is at most $\lStroke(l)$; see
Fig.~\ref{fig:workflow}(a). The hull of a polyline is then the union of its segments' hulls. We approximate hulls by simple polygons.

\textsc{Step 1 -- Identification.}  For each road name $n$, each color
$c$ and each font size $f$ we define the intersection graph of the
hulls of the line segments $L_{n,c,f}=\{l\in L \mid \lName(l)=n,\ \lColor(l)=c \text{ and } \lFont(l)=f\}$. In this
intersection graph each hull is a vertex and two vertices are
connected if and only if the corresponding hulls intersect.  In each
(non-empty) intersection graph we identify all connected components,
which we call \emph{road components}; e.g., in
Fig.~\ref{fig:workflow}(a) the blue segments form a road
component. Thus, based on $L$ we obtain a set~$\mathcal C$ of road
components. By definition, each component $C\in \mathcal C$ has a
unique name~$\lName(C)$, stroke width~$\lStroke(C)$, color
$\lColor(C)$ and font size~$\lFont(C)$.

\textsc{Step 2 -- Simplification.} 
For each road component~$C\in \mathcal C$ we geometrically form the
union of the corresponding hulls. Thus, the result is a simple
polygon~$P$ (possibly with holes); see Fig.~\ref{fig:workflow}(b),
top. This polygon describes the contour of the road component as drawn on the map. We discard all polygons whose area is smaller
than some threshold as they are too small to be labeled; we use the
area of the letter \texttt{W} as threshold. For each remaining polygon $P$ we
construct the \emph{skeleton} of $P$ as a linear representation of the
corresponding road component such that labels centered on the skeleton
are guaranteed to be contained in $P$.  This skeleton is based on the
conforming Delaunay triangulation of the interior of~$P$ following
Bader and Weibel~\cite{Bader97}. For triangles that have one or three
\emph{internal} edges, i.e., edges that do not belong to the boundary
of~$P$, we connect the triangle centroid to the midpoints of the
internal edges. For triangles with two internal edges, we simply
connect the midpoints of these two edges, see
Fig.~\ref{fig:workflow}(b), bottom. From those line segments, we form
a set of maximal polylines by appending all those line segments that
meet at the midpoint of a triangle edge (but not at a triangle
centroid). Since these polylines may consist of many vertices and meander \emph{locally}, we simplify them using the Douglas-Peucker
algorithm, but only if the simplified shortcuts keep a distance of at
least $\lFont(C)/2$ to the boundary of~$P$, see
Fig.~\ref{fig:workflow}(c).  Finally, we delete any segment~$l$ whose
text box $B_l$ is not completely contained in~$P$. Here the \emph{text
  box} $B_l$ of $l$ is defined as a rectangle centered at~$l$ with two
sided parallel to~$l$. These parallel sides have the same length as~$l$, the two orthogonal sides have length~$\lFont(C)$, see
Fig.~\ref{fig:workflow}(b), bottom.  Segments with the text box not
contained in $P$ may occur at the protrusions of the component where
circular arcs are approximated by polylines, see
Fig.~\ref{fig:workflow}(c), top left.  The remaining set of polylines
forms the skeleton of $P$.

Thus, for each road component~$C$ we obtain a skeleton such that all text
boxes of the skeleton edges are contained in $P$. This resolves~\ILanes. We annotate each skeleton edge with the name, stroke width, color and font size of~$C$.

\textsc{Step 3 -- Planarization.} So far polylines of different road
components may intersect at other points than their end points, e.g.,
polylines representing bridges and tunnels may cross other
polylines. As motivated in Sect.~\ref{sec:issues}, we subdivide these
polylines to resolve intersections; see
Fig.~\ref{fig:workflow}(d). More precisely, if two line segments
$\overline{pq}$ and $\overline{rs}$ of two polylines intersect at a point~$t$, we
replace them by the four segments $\overline{pt}$, $\overline{tq}$,
$\overline{rt}$ and $\overline{ts}$. We do the intersection tests
with a certain tolerance to identify $T$-crossings
safely. However, this may yield short stubs that protrude
junctions slightly; we remove those stubs. Thus, this step resolves
\ICrossings and yields a set of annotated
 polylines only intersecting in vertices.

 \textsc{Step 4 -- Transformation.} 
 Next we create the abstract road
 graph from the polylines of the previous step. 
 As a result of Step 3 we know that any two polylines intersect only in vertices. We first take the union of all polylines, identify vertices that are common to two or more polylines and mark these vertices as \emph{junction seeds}. This induces already a planar graph $G=(V,E)$ with polyline edges whose vertices $V$ are either junction seeds or have degree 1. 
 It remains to partition the edges of $G$ into road sections and junction edges. Initially, we mark all edges as road sections. 
 We distinguish two types of junction seeds in $G$.
 
 If a junction seed $v$ has degree at least $3$, only two of its incident edges $e$ and $e'$ belong to the same road $R$ and all other incident edges belong to different roads (and have a different road type than $R$) then we do not create any junction edges at $v$, see Fig.~\ref{fig:workflow}(e), small box.
 Since $R$ is the only road that may use the junction at $v$ and it is visually clear that all other roads end at $v$ we can safely treat $v$ as an internal vertex of a road section of $R$.
 So we disconnect all incident edges of $v$ except $e$ and $e'$ from $v$ and let each of them end at its own slightly displaced copy of $v$. The edges $e$ and $e'$ are merged at $v$ and the new edge remains a road section. This resolves the situation as desired.
 
 For all other junction seeds we create junction edges as follows. Let $v$ be a junction seed and let $E_v$ be the set of edges incident to $v$. We intersect the hulls of all  edges in $E_v$ and project their intersection points onto the corresponding edges, see Fig.~\ref{fig:workflow}(f). For each edge $e \in E_v$ we determine the projection point~$p_e$ that is farthest away from $v$ (in geodesic distance). If the distance between~$p_e$ and $v$ exceeds a given threshold $\delta$, we shift $p_e$ to the point on $e$ that has distance $\delta$ from $v$. Now we subdivide $e$ at~$p_e$ and mark the edge $\{v,p_e\}$ as a junction edge; the other edge at $p_e$ (if non-empty) remains a road section. The threshold $\delta$ ensures that roads running closely in parallel are not completely marked as junction edges. Figure~\ref{fig:workflow}(g) shows the resulting abstract road graph. 

 To resolve \ILong we
 subdivide road sections whose length exceeds a certain threshold (in
 our experiment 350 pixels) by inserting a very short junction edge.

\textsc{Step 5 -- Resolving Overlaps.} By Step 2 the
hulls of edges that belong to the same road component
do not overlap. However, if two sections of different roads run
closely in parallel, their hulls (and hence their labels) may overlap. 
We identify overlaps of
the hulls of non-incident edges in $G$ and block the corresponding
parts of the edge whose road is less
important for placing labels; 
ties are broken arbitrarily. More complex approaches
using road displacement could be applied, however, 
we have chosen a simple solution.  By design hulls
of incident edges may only overlap if both are junction edges; those
overlaps are handled by the labeling algorithms; see
Sect.~\ref{sec:labeling}.  This resolves \IFatLabels.


\section{Phase 2 -- Label Placement in Road
  Graphs}\label{sec:labeling} 
  
  In this section we present the four different methods for solving \MaxTotalCovering that we subsequently evaluate in our experiments in Sect.~\ref{sec:evaluation}. Furthermore, we describe a technique for decomposing road graphs into several smaller, independent components that may speed up computations.

\subsection{Labeling Methods}

\myparagraph{\GreedyAlgo.} An obvious base-line heuristic to obtain lower bounds is to simply place a well-shaped label on each individual road section that is long enough to admit such a label without extending into any junctions. We use this approach to show that it is beneficial to position labels across junctions.
 
\myparagraph{\textsc{Mapnik}.} Mapnik (\url{http://mapnik.org}) is a standard
open source renderer for OpenStreetMap that includes an road labeling algorithm.
The algorithm iteratively labels so-called \emph{ways}, which are
polylines describing line features in OpenStreetMap. Along each way it
places labels with a certain spacing and locally ensures that labels
do not intersect already placed labels of other ways.  It does not use any semantic structure from the road network (e.g., road sections), but
relies on how the contributors of OpenStreetMap modeled single ways. We may run the rendering algorithm and extract all placed labels from its output. 

\myparagraph{\TreeAlgo.}
The tree-based heuristic makes use of our recently proposed algorithm 
that optimally solves \MaxTotalCovering if $G$ is a tree~\cite{rlTheory}. The basic idea for trees is that a
placed label splits the tree into several independent sub-trees, which then are
labeled recursively. Using dynamic programming we reuse already
computed results so that the algorithm's complexity becomes
polynomial, namely $O(|V|^5)$ running time and $O(|V|^2)$ space. Applying
some further intricate modifications we improved this to $O(|V|^3)$ time
and $O(|V|)$ space, and $O(|V|^2)$ time if each road in $G$ is a path. We omit the details and use that algorithm as a \emph{black
  box}. If~$G$ is a tree, our heuristic optimally labels~$G$. Otherwise it computes
a spanning tree~$T$ on~$G$ using Kruskal's algorithm and computes an
optimal labeling for $T$. We
construct~$T$ such that all road sections of~$G$ are contained
in~$T$. Since a road section is only incident to junction edges, this
is always possible. In
Sect.~\ref{sec:evaluation} we show that large parts of realistic road
networks can actually be decomposed into paths and trees without losing
optimality.

\myparagraph{\ILPAlgo.} In order to provide upper bounds for the
evaluation of our labeling algorithms, we implement a mixed-integer linear programming (MILP) model that solves \MaxTotalCovering optimally on arbitrary abstract road
graphs. The basic idea is to discretize all possible label positions and to
restrict the space of feasible solutions to non-overlapping sets of
labels.

We now describe the MILP formulation in detail. To simplify the
presentation, we drop the rather technical concept of
\emph{well-shaped} labels, but note that it can be easily incorporated
into the MILP. In the following let the edges of $G$ be (arbitrarily)
directed.

We first discretize the problem as follows. Two labels $\ell$ and
$\ell'$ are \emph{equivalent} if they cover the same edges in the same
order, i.e., $\mathcal P_\ell=\mathcal P_{\ell'}$, and only their end points differ; see
Fig.~\ref{fig:rules}(b).  For
each such equivalence class we create one label~$\ell$; we denote its
equivalence class by~$\mathcal K_\ell$. Further, let~$L$ denote the
set of such created labels.  The main idea of the MILP is to
select a subset of~$L$ and to determine the exact positions of the
labels' end points on their terminals such that they do not overlap
and label a maximum number of road sections.

Now, consider a label~$\ell\in L$ and the path $\mathcal
P_\ell=(e_1,e_2,\cdots,e_{k-1},e_k)$ that is covered by~$\ell$; see
Fig.~\ref{fig:rules}(b). In the following we call $e_1$ and $e_k$
the \emph{terminals} of~$\ell$ and the others \emph{internal edges}
of~$\ell$.  Assume that the head of the label~$\ell$ lies on $e_1$ and
the tail on $e_k$, then $\ell$ can slide along~$\mathcal P_\ell$
changing the covered road sections until the head or tail of~$\ell$
\emph{hits} an end point of $e_1$ or $e_k$, respectively. At each
position, $\ell$ coincides with an equivalent
label~$\ell'$. Obviously, those labels exactly form $\mathcal
K_\ell$. Further, there exist two positions on~$e_1$ such that the
head of $\ell$ has either minimum geodesic distance $a$ or maximum
geodesic distance $b$ to the source of $e_1$, respectively. We define
the interval~$H_\ell=[a,b]$. Analogously, we define the interval
$T_\ell$ for the tail of $\ell$ and the edge $e_k$.

 For each label~$\ell \in L$ we introduce
the variables~$x_\ell\in\{0,1\}$,~$h_{\ell}\in H_\ell$
and~$t_{\ell}\in T_\ell$, and for each road section $e\in E$ the
variable $y_e\in\{0,1\}$.  We interpret $x_\ell=1$ such
that~$\ell$ is selected for the labeling. The variables~$h_{\ell}$
and $t_{\ell}$ are interpreted as the geodesic distances of the head
and tail to the source of the head's and tail's terminal, respectively; see
Fig.~\ref{fig:rules}(c). We interpret $y_e=1$ as road section~$e$
being labeled and 
maximize the sum $\sum_{e\in E}y_e$ subject to the following constraints.

For each $\ell\in L$ we require
\begin{align}\label{constr:length}
 \covered(e_1,\ell) + \length(e_2)+\ldots+\length(e_{k-1})+\covered(e_k,\ell) = \length(\ell),
\end{align}
where $\mathcal P_{\ell}=(e_1,\dots,e_k)$, $\length(\ell)$ denotes the given
length of~$\ell$ and $\covered(e,\ell)$ is a linear expression
describing what length of $e$ is covered by $\ell$. This expression
depends on which end point of $e$ is covered, whether the head or tail
of~$\ell$ lies on~$e$, and on the position variables $h_\ell$ and $t_\ell$,
respectively; we omit the technical definition.
Further, for each pair $\ell',\ell \in L$ we require
\begin{align}
  x_\ell+x_{\ell'}\leq 1  & \quad\text{ if an edge of $\ell$ is an internal edge of $\ell'$ }\label{constr:overlap1}\\
 \begin{split}
  h_\ell-h_{\ell'} \leq M(2-x_\ell-x_{\ell'}) & \quad\text{ if the heads of~$\ell$ and $\ell'$ lie on a common}\label{constr:overlap2}\\[-.75ex]
  & \quad\text{ terminal~$e$ and $\ell$ covers the source of~$e$.}
 \end{split}
\end{align}
For each road section $e\in E$ and all labels $\ell_1,\dots,\ell_k\in L$ labeling~$e$ we require
\begin{align}
 y_e  \leq x_{\ell_1}+\cdots+x_{\ell_k}\label{constr:counting}
\end{align}
Constraint~(\ref{constr:length}) ensures that each label has the
desired length~$\length(\ell)$. Constraint~(\ref{constr:overlap1})
ensures that a label does not overlap another label internally, i.e.,
it (partly) covers an edge that is completely covered by another
label. Constraint~(\ref{constr:overlap2}) ensures that labels ending
on the same road section do not overlap on that edge, but~$\ell$ ends
on $e$ before~$\ell'$ starts; see Fig.~\ref{fig:rules}(c). Similar
constraints are introduced for the other combinations on how heads and
tails of $\ell$ and $\ell'$ can lie on a common edge, and on whether
source or target of~$e$ is covered. For an appropriate large $M$ the
constraint is trivially satisfied if $\ell$ or $\ell'$ is not selected
for the labeling.  Finally, Constraint~(\ref{constr:counting}) ensures
that road section~$e$ is only counted as labeled, if there is at least
one selected label covering~$e$.

Since $L$ models all possible label positions and the constraints restrict the space of feasible solutions to non-overlapping sets of labels, it is clear that any optimal solution of the above MILP corresponds to an optimal solution of \MaxTotalCovering.

\begin{theorem}
\ILPAlgo solves \MaxTotalCovering optimally.
\end{theorem}

Finding an optimal solution for a MILP formulation is NP-hard in
general and remains NP-hard for the stated formulation, because
\MaxTotalCovering is NP-hard. However, it turns out that in practice we can
apply specialized solvers to find optimal solutions for reasonably sized instances in acceptable time, see Sect.~\ref{sec:evaluation}.

\subsection{Decomposition of Road Networks.}
\begin{figure}[t]
\centering
\includegraphics[page=1,scale=1.2]{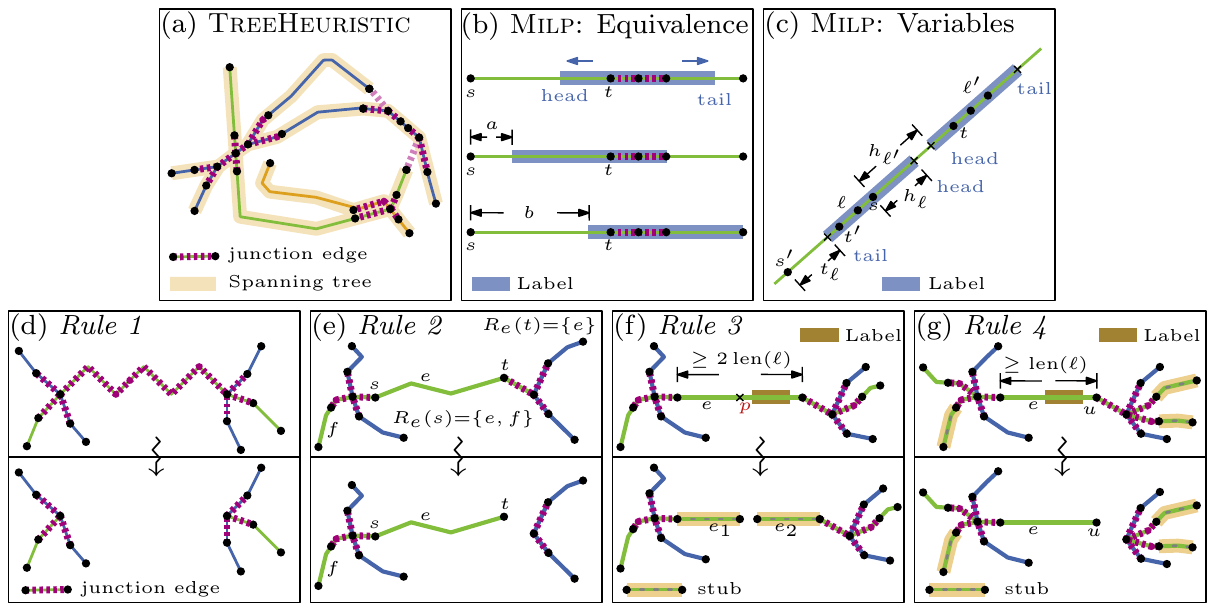}
\caption{Illustration of the algorithms. Edges of the same color belong to the same road.}
\label{fig:rules}
\end{figure}

We may speed up both our heuristic \TreeAlgo and the exact approach \ILPAlgo by decomposing the road graph into smaller, independent components to be labeled separately, i.e., components whose individual optimal solutions compose to a conflict-free optimal solution of the initial road graph.  Such a decomposition allows to compute solutions in parallel with either of the above methods and it further decreases the total combinatorial complexity. The decomposition rules guarantee that the labelings of the components can always be merged without creating any label overlaps.
We name this technique \Shredder.

\textsc{Step 1 -- Decomposition.} For many road sections, e.g., long
sections, of real-world road networks labels can be easily placed
preserving the optimal labeling.  We iterate through the edges of $G$ and cut or remove some of them if one of the following rules applies. As a result the graph decomposes into independent connected
components; see Fig.~\ref{fig:rules}(d)--(g). Let~$e$ be the
currently considered edge and let~$R$ be the road of $e$.

\RuleA. If $e$ is a junction edge and it cannot be completely covered
by a well-shaped label, i.e., $e$ is not well-shaped, then remove $e$.

\RuleB. If $e$ is a road section that ends at a junction that is not
connected to any other road section of~$R$, then detach $e$ from that junction.

\RuleC. If $e$ is a road section, a well-shaped label~$\ell$
fits on~$e$, and $e$ is at least twice as long as~$\ell$, then cut $e$ at its
midpoint.

\RuleD. If $e$ is a road section, a well-shaped label~$\ell$ fits on $e$, and
$e$ is connected to a junction that is only connected to road sections
of~$R$ that may completely contain a well-shaped label, then detach $e$ from that junction.





On each edge we apply at most one rule. If we apply
\RuleC or \RuleD on an edge~$e$, we call $e$ a \emph{long-edge}. Afterwards, we determine all connected components
of the remaining graph~$G'$, which are then independently labeled. 

\textsc{Step 2 -- Label Placement.} For the constructed components we compute solutions in parallel with either of the above methods.

\textsc{Step 3 -- Composition.} Finally, we compose the labelings of
the second step to one labeling. Due to the decomposition, no two
labels of different components can overlap. If a long-edge~$e$ is not
labeled, we place a label on it, which is possible by definition.  We
adapt the algorithms of Step 2 such that they do not count labeled
road sections that were created by \RuleC, but we count the
corresponding long-edge in this step.

\paragraph{Correctness.} We now prove the correctness of the
approach. To that end we first formalize the presented rules. We
assume that the edges of~$G$ are (arbitrarily)
directed. 

\RuleA. If $e$ is a junction edge and it cannot be completely covered
by a well-shaped label, i.e., $e$ is not well-shaped, then remove $e$.

\RuleB. Let~$R_e(v)$ be the set of road sections that belong to the
same road as~$e$, and that are reachable from~$v$ in~$G$ when only
traversing junction edges. If $e=(s,t)$ is a road section and
$R_e(u)=\{e\}$ for an $u\in\{s,t\}$, then remove the junction edge incident
to~$u$.

\RuleC. If $e=(s,t)$ is a road section, a well-shaped label~$\ell$
fits on~$e$, and $e$ is twice as long as~$\ell$, then replace $e$ by
the road sections~$e_1=(s,u_1)$ and $e_1=(u_2,t)$, where $u_1$ and
$u_2$ are two new vertices at the midpoint~$p$ of $e$, $e_1$ is a sub-polygon of~$e$ from
$s$ to $u_i$ and $e_2$ is a sub-polygon of~$e$ from $u_2$ to $t$. We
mark $e_1$ and $e_2$ as \emph{stubs} and call $e$ a \emph{long-edge}.

\RuleD.  If $e=(s,t)$ is a road section, a well-shaped label~$\ell$
fits on~$e$ and for at least one end node~$u\in\{s,t\}$ the road
sections in $R_e(u)\setminus\{e\}$ are all stubs, then remove the
junction edge incident to~$u$. We mark $e$ as \emph{stub} and call $e$
a \emph{long-edge}.

\newcommand{\thmRules}{Let $G$ be an
  abstract road graph and let $\mathcal L$ be the resulting labeling
  after applying \Shredder\ combined with an algorithm that yields optimal labelings. An optimal labeling~$\mathcal L'$ of $G$ and $\mathcal
  L$ label the same number of road sections.  }

\begin{theorem}\label{apx:thm:rules}
\thmRules
\end{theorem}

\begin{proof}
  Let~$G=(V,E)$ be an abstract road graph and let $\mathcal L'$ be an
  optimal labeling of $G$, i.e., no more road sections can be labeled.
  We show that we can transform $\mathcal L'$ into a labeling
  $\mathcal L$ that is found by \Shredder, and, furthermore, $\mathcal
  L$ and $\mathcal L'$ label the same number of road sections. If not
  mentioned otherwise, we assume a label to be well-shaped.
 
  \RuleA. Assume that we apply \RuleA on $G$ by deleting a junction
  edge~$e$ that cannot be completely covered by a well-shaped
  label. By definition no label may end on a junction edge, but it
  must end on a road section. Thus, in any labeling the edge $e$
  cannot be covered by any label. We therefore can delete the edge
  preserving the optimal labeling, i.e., an optimal labeling of $G$
  and $G'=(V,E\setminus\{e\})$ label the same number of road sections.

  \RuleB. Assume that we apply \RuleB on the edge~$e$. Since $e$ is
  the only edge in $R_e(u)$, the edge is the end of a road, i.e., all
  other edges incident to $u$ cannot belong to the same road of~$e$. Since $e$ is a junction edge, no label may end on a junction
  edge, and labels may only cover edges of the same road, no label can
  cover~$e$ in any labeling. We therefore can delete the edge preserving the optimal labeling.

  \RuleC. Assume that we apply \RuleC on the road section~$e=(s,t)$
  splitting~$e$ into the edges~$e_1=(s,u_1)$ and $e_2=(u_2,t)$.
  Since $e$  may
  contain a well-shaped label, the road section $e$ must be labeled in
  $\mathcal L'$.

  If~$e$ is only labeled by labels that are completely
  contained in $e$, i.e., they do not cover other edges of~$G$, we
  will find one of those labels in the composition step of \Shredder.

  Hence, assume that there is a label~$\ell_1 \in \mathcal L'$ that
  covers $e$ and $s$. Since $e$ is twice as long as the label length
  of~$e$, this label cannot cover the location of $u_1$($u_2$). The
  same applies for a label $\ell_2\in \mathcal L$ that covers~$e$
  and~$t$. Since $e$ is labeled by $\ell_1$ ($\ell_2$) we can remove
  all other labels that only label~$e$ without changing the maximum
  number of labeled road sections. Hence, the point at $u_1$ is not
  covered by any label, which means we can split $e$ at this
  point preserving the optimal labeling.

  \RuleD. Assume that we apply \RuleD on the road section~$e=(s,t)$
  with $u=s$; same arguments hold for $u=t$. Hence, the road sections
  in $R_e(u)\setminus\{e\}$ are all stubs, i.e., well-shaped labels can placed
  on any of these road sections. Let~$j$ be the junction edge that is
  connected to~$s$. Assume that there is a label~$\ell$ that labels
  $e$ and an edge~$e'$ of $R_e(u)\setminus\{e\}$ such that $u$ is covered~ by~$\ell$.

  If~$e$ and $e'$ are also labeled by other labels, we can
  remove~$\ell$ without changing the number of labeled road sections
  and remove~$j$.  So assume that $e$ is not labeled by another
  label. In that case we remove~$\ell$ and place a label that
  completely lies on $e$ without covering any other edges; by
  definition of the rule this is possible. If~$e'$ is also not labeled
  by any other label, we also place a label on $e'$, which is
  possible, because~$e'$ is a stub. Hence, we can remove~$j$ preserving the optimal labeling. 
\end{proof}

\section{Evaluation}\label{sec:evaluation}

We evaluate our framework and in particular the performance of our new tree-based labeling heuristic by conducting a set of experiments on the road networks of 11 North American and European cities; see~Table~\ref{table:instances16}.
While the former ones are characterized by
grid-shaped road networks, the latter ones rarely posses such
regular geometric structures. 
Since the road networks in rural areas are
much sparser than those of cities, we refrained from considering these networks and focused on the more complex city road networks.  
We extracted the abstract road graphs from the data provided by
OpenStreetMap\footnote{\url{http://www.openstreetmap.org}}. We applied the spherical
Mercator projection ESPG:3857, which is also known as \emph{Web
  Mercator} and used by several popular map-services.  We considered
the three scale factors 4.773, 2.387 and 1.193, which approximately
correspond to the map scales 1:15000, 1:8000, 1:4000\footnote{\url{http://wiki.openstreetmap.org/wiki/Zoom\_levels}}. Further, they correspond to the
\emph{zoom levels} 15, 16 and 17, respectively, which are widely used
by map services as OpenStreetMap. Those zoom levels show road networks
in a size that already allows labeling single road sections, while the
map is not yet so large that it becomes trivial to label the
roads. We applied the standard drawing style for OpenStreetMap, which
in particular includes the stroke width and color of roads as well
as the font size of the labels. Further, this specifies for each zoom
level the considered road categories; the higher the zoom level the
more categories are taken into account.

Our implementation is written in C++ and compiled with GCC 4.8.4 using
optimization level~\texttt{-O3}. MILPs were solved by
Gurobi\footnote{\url{http://www.gurobi.com}} 6.0. The experiments were performed on a 4-core
Intel Core i7-2600K CPU clocked at 3.4 GHz, with 32 GiB RAM. The
\Shredder-approach labels single components in parallel. For computing
the Delaunay triangulation we used the library
Fade2d\footnote{\url{http://www.geom.at}}.

For each city and each zoom level we applied the algorithms
\GreedyAlgo, \TreeAlgo, \Shredder+\TreeAlgo, \ILPAlgo and
\Shredder+\ILPAlgo. We adapted the algorithm such that short road
sections (shorter than the width of the letter \texttt{W}) are not
counted, because they are rarely
visible. Further, we let Mapnik (Version 3.0.9) render the same
input. For each label we identified for each of its letters the
closest road section~$r$ with the same name and counted it as
labeled. Since Mapnik does not optimize the labeling by the same
criteria as we do, we compensate this by also counting neighboring road sections as labeled if the junction in between them is not incident to any other road section. This accounts for those long road sections that we split artificially to resolve~\ILong.

\begin{table}[t]
  \caption{Statistics for Baltimore (BA), Berlin (BE), Boston (BO), Los Angeles (LA), London (LO), Montreal (MO),   Paris (PA), Rome (RO), Seattle (SE), Vienna (VI) and Washington (WA) for zoom level 15, 16 and 17. \emph{OSM}: Number of input segments in thousands. \emph{Segm.:} Percentage of segments after Phase 1, Step 3 in relation to input segments.
    \emph{Graph}: Number of road sections after Phase 1 in thousands. \emph{Time:} Running time for Phase 1.
  }
\label{table:instances16}
\label{table:instances15}
\label{table:instances17}
\centering
\small
\begin{tabular}{llccccccccccc}
\toprule
 && BE & LO & PA & RO & VI & BA & BO & LA & MO & SE & WA \\
\midrule
 \parbox[t]{3mm}{\multirow{4}{*}{\rotatebox[origin=c]{90}{Zoom 15}}} &OSM & 143.9 & 437.6 & 225.1 & 87.7 & 85.1 & 196.1 & 174.5 & 257.1 & 134.6 & 315.3 & 82.2\\
&Segm.& 62 & 80 & 65 & 66 & 63 & 52 & 54 & 74 & 78 & 70 & 39\\
&Graph & 28.5 & 78.5 & 35.3 & 10.3 & 14.8 & 24.7 & 20.1 & 61.3 & 31.9 & 63.1 & 8.7\\
&Time & 16 & 62 & 28 & 10 & 10 & 22 & 19 & 42 & 20 & 40 & 8\\
\midrule
\parbox[t]{3mm}{\multirow{4}{*}{\rotatebox[origin=c]{90}{Zoom 16}}}
&OSM & 225.0 & 563.4 & 292.5 & 117.0 & 119.9 & 332.1 & 225.0 & 327.0 & 161.4 & 433.1 & 103.9\\
&Segm. & 55 & 73 & 62 & 62 & 54 & 40 & 50 & 67 & 72 & 59 & 37\\
&Graph & 37.9 & 105.4 & 49.9 & 15.4 & 18.9 & 33.8 & 27.8 & 80.6 & 40.2 & 77.1 & 11.4\\
&Time & 21 & 65 & 32 & 12 & 11 & 28 & 21 & 44 & 21 & 42 & 9\\
\midrule
\parbox[t]{3mm}{\multirow{4}{*}{\rotatebox[origin=c]{90}{Zoom 17}}}
&OSM & 225.0 & 563.4 & 292.5 & 117.0 & 119.9 & 332.1 & 225.0 & 327.0 & 161.4 & 433.1 & 103.9\\
&Segm. & 64 & 80 & 69 & 70 & 60 & 46 & 56 & 73 & 83 & 64 & 43\\
&Graph & 47.1 & 127.0 & 59.1 & 19.4 & 22.3 & 39.5 & 32.3 & 90.4 & 47.4 & 87.9 & 13.0\\
&Time & 24 & 67 & 33 & 13 & 11 & 29 & 22 & 46 & 22 & 43 & 10\\
\bottomrule
\end{tabular}
\end{table}

The raw data of our experiments is made available on
\href{http://i11www.iti.kit.edu/roadlabeling}{\texttt{i11www.iti.kit.edu/roadlabeling}}. On this page we also
provide interactive maps of the cities Berlin, London, Los Angeles and
Washington, which present the computed labelings.

\begin{table}[t]
  \caption{\emph{Speedup}: Ratio of running times of two algorithms. \emph{Quality:} Ratio of the number of labeled road sections computed by two algorithms. }
\label{table:speedup}
\centering
\small
\begin{tabular}{lccccccccccccc}
\toprule
 & Ratio & BE & LO & PA & RO & VI & BA & BO & LA & MO & SE & WA & \textbf{Avg.}\\ \midrule
\parbox[t]{1mm}{\multirow{3}{*}{\rotatebox[origin=c]{90}{Speedup\hspace{1.5ex}}}}
 & $\frac{\ILPAlgo}{\Shredder+\ILPAlgo}$ & 3.44 & 3.07 & 2.51 & 1.71 & 3.12 & 1.44 & 2.33 & 1.3 & 1.79 & 3.1 & 1.32 & \textbf{2.29}\vspace{1ex}\\
 & $\frac{\TreeAlgo}{\Shredder+\TreeAlgo}$ & 1.77 & 1.8 & 1.73 & 1.62 & 1.71 & 1.57 & 1.71 & 1.37 & 1.75 & 1.68 & 1.35 & \textbf{1.64}\vspace{1ex}\\
 & $\frac{\Shredder+\ILPAlgo}{\Shredder+\TreeAlgo}$ & 2.82 & 2.32 & 3.33 & 2.54 & 2.74 & 6.84 & 3.06 & 21.59 & 6.36 & 5.32 & 10.59 & \textbf{6.14}\vspace{1ex}\\
\midrule
 \parbox[t]{1mm}{\multirow{5}{*}{\rotatebox[origin=c]{90}{Quality\hspace{5ex}}}}
  & $\frac{\Shredder+\TreeAlgo}{\TreeAlgo}$ & 1.01 & 1.0 & 1.0 & 1.0 & 1.01 & 1.01 & 1.0 & 1.01 & 1.02 & 1.01 & 1.02 & \textbf{1.01}\vspace{1ex}\\
 & $\frac{\Shredder+\TreeAlgo}{\ILPAlgo}$ & 1.0 & 1.0 & 0.99 & 0.99 & 0.99 & 0.96 & 0.99 & 0.96 & 0.97 & 0.97 & 0.91 & \textbf{0.97}\vspace{1ex}\\
 & $\frac{\text{Mapnik}}{\ILPAlgo}$ & 0.74 & 0.85 & 0.83 & 0.91 & 0.76 & 0.71 & 0.8 & 0.62 & 0.61 & 0.8 & 0.68 & \textbf{0.75}\vspace{1ex}\\
 & $\frac{\GreedyAlgo}{\ILPAlgo}$ & 0.58 & 0.49 & 0.4 & 0.38 & 0.48 & 0.39 & 0.42 & 0.39 & 0.46 & 0.37 & 0.24 & \textbf{0.42}\vspace{1ex}\\
 & $\frac{\Shredder+\TreeAlgo}{\text{Mapnik}}$ & 1.36 & 1.19 & 1.2 & 1.09 & 1.29 & 1.37 & 1.25 & 1.55 & 1.58 & 1.21 & 1.33 & \textbf{1.31}\vspace{1ex}\\
\bottomrule
\end{tabular}
\end{table}

\begin{figure}[t]
\centering
\subfigure[]{\includegraphics[scale=0.8]{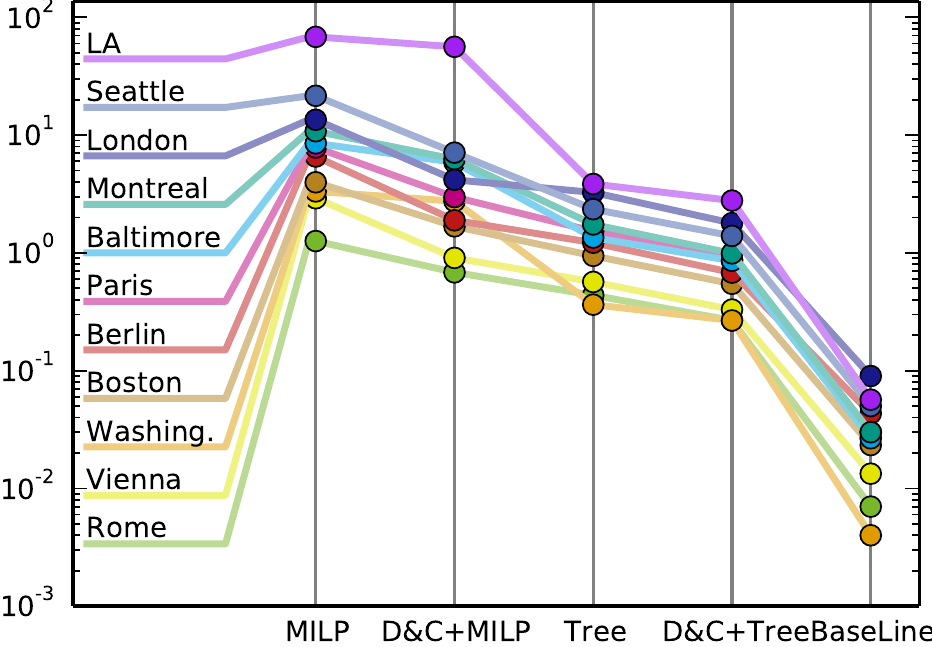}\label{plot:all:time}}
  \centering \subfigure[]{\includegraphics[scale=0.8]{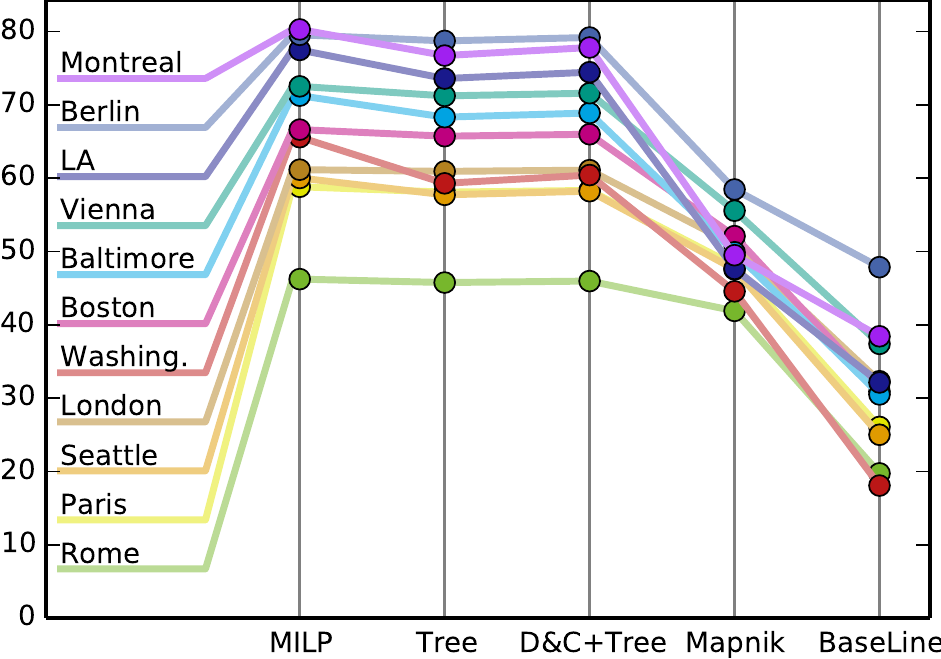}\label{plot:all:rs}}

  \caption{ 
 \subref{plot:all:time}  Running times in seconds of the
    algorithms (logarithmic scale). 
 \subref{plot:all:rs} Percentage of labeled road sections over all  zoom levels broken into the different  algorithms. 
}\vspace{-2ex}
\end{figure}

\textit{Phase 1.}   With a maximum of 67 seconds (London,
zoom 17) and 27 seconds averaged over all instances, Phase 1 can be
applied on large instances in reasonable time. During Phase 1 the number of segments
is reduced to between $40 \%$ and $83 \%$ of the original instance (measured after Step 3, before creating junction edges); see
Table~\ref{table:instances15}. This clearly indicates that the
procedure aggregates many lanes, since by design the approach does not
change the overall geometry, but the simplification maintains the shape of the original network. This is also confirmed by the
labelings; see Fig.~\ref{fig:motivation}(b)--(c) and interactive maps.

\textit{Phase 2, Running Time.} We first consider the average
running times over all zoom levels; see Fig.~\ref{plot:all:time}. We
did not measure the running times of Mapnik, because its labeling
procedure is strongly interwoven with the remaining rendering
procedure, which prevents a fair comparison. As to be expected \ILPAlgo is the slowest method (max.\ 126
sec., Los Angeles, ZL 15), while \GreedyAlgo is the fastest procedure
(max.\ 0.17 sec.). Combining \ILPAlgo with \Shredder\ results in an
average speedup of 2.29 over all instances and a maximum speedup of
3.44; see Table~\ref{table:speedup}. 

The algorithm \TreeAlgo needs less than 4.7 seconds and its median is
about 1.3 seconds.  Hence, despite its worst-case cubic asymptotic
running time, it is fast in practice. Similar to \ILPAlgo, it is
further enhanced by combining it with~\Shredder\ for a speedup of 1.64
with respect to~\TreeAlgo, and an average speedup of 6.14 with respect
to \Shredder+\ILPAlgo; see Table~\ref{table:speedup}. In the latter
case it has even a maximum speedup of about 21.6.  Since decomposing
and composing the labelings is done sequentially, the theoretically
possible speed up using~\Shredder\ is not achieved.
 

\begin{figure}[t]
\centering
\includegraphics[scale=0.8]{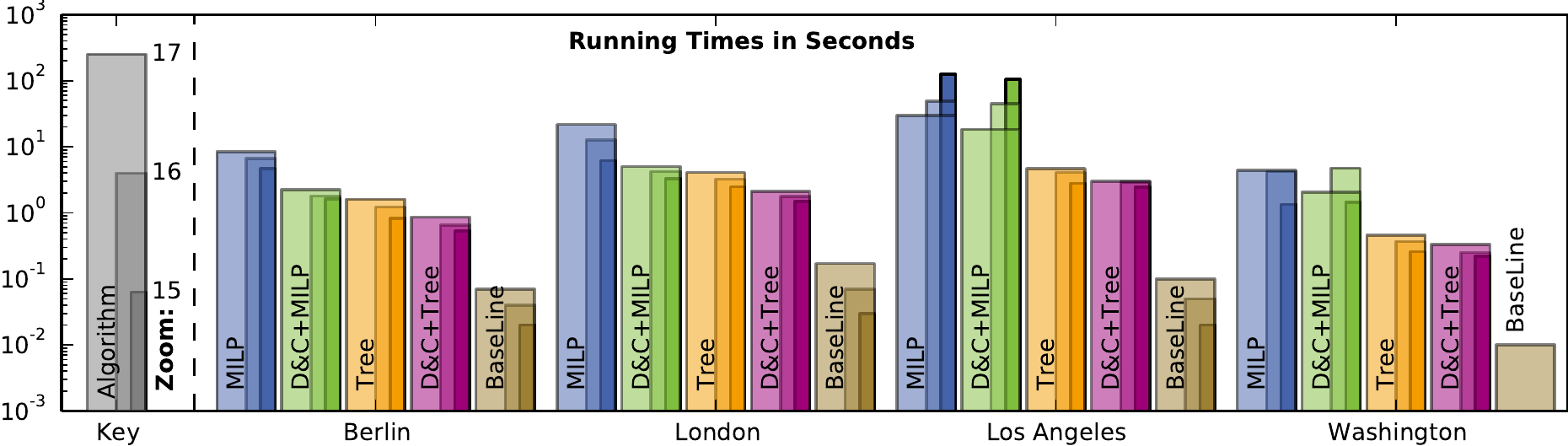}
\caption{Running times of the algorithms broken down in zoom levels and algorithms. The width of the bars (thin, mediate, wide) corresponds with zoom levels (15, 16, 17). 
}
\label{fig:running-time:selection}
\end{figure}

\begin{figure}[t]
\centering
\includegraphics[scale=0.8]{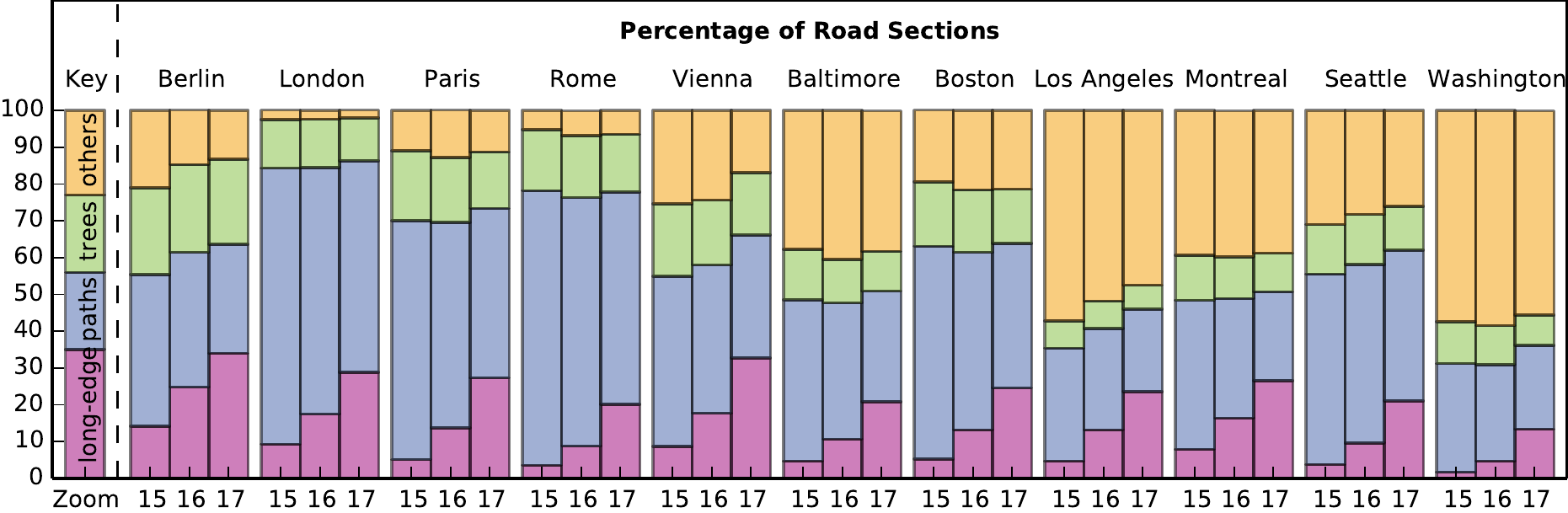}
\caption{
Decomposition of road networks by \Shredder. Percentage of long-edges, road sections in paths, trees and other components in which the road networks is decomposed.   }
\label{fig:components}
\end{figure}

If we break down the running times into single zoom levels, we observe
similar results; see e.g.,
Fig.~\ref{fig:running-time:selection}. Since with increasing zoom level the instance size grows, for most of the algorithms also the
running time increases. Only for North American cities and \ILPAlgo we
observe that the running time for instances of smaller zoom levels are
higher than for larger zoom levels.

\begin{figure}[t]
\centering
\subfigure[]{\includegraphics[scale=0.8]{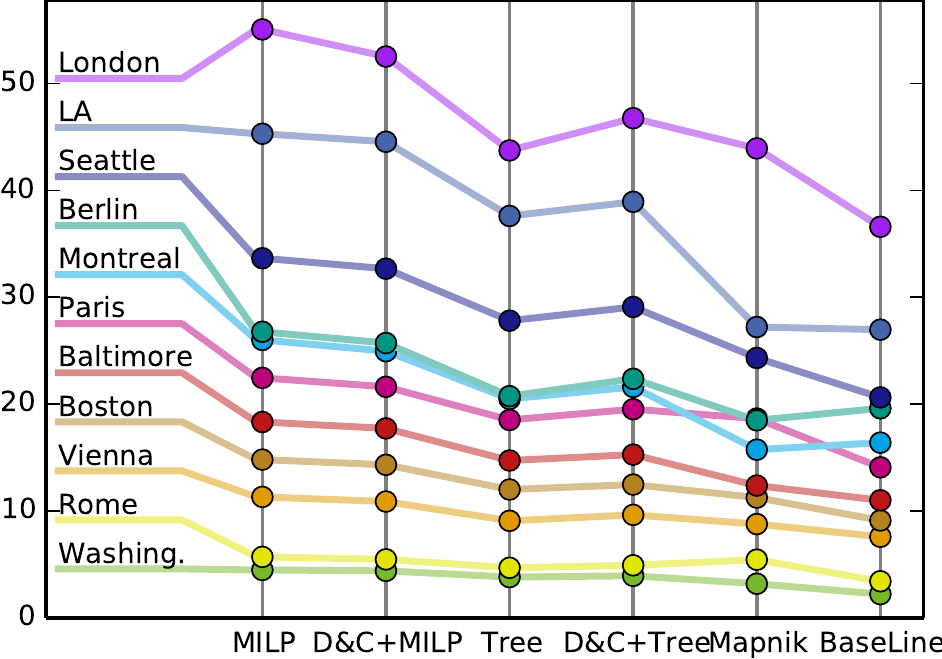}\label{apx:plot:labels}}
\subfigure[]{\includegraphics[scale=0.8]{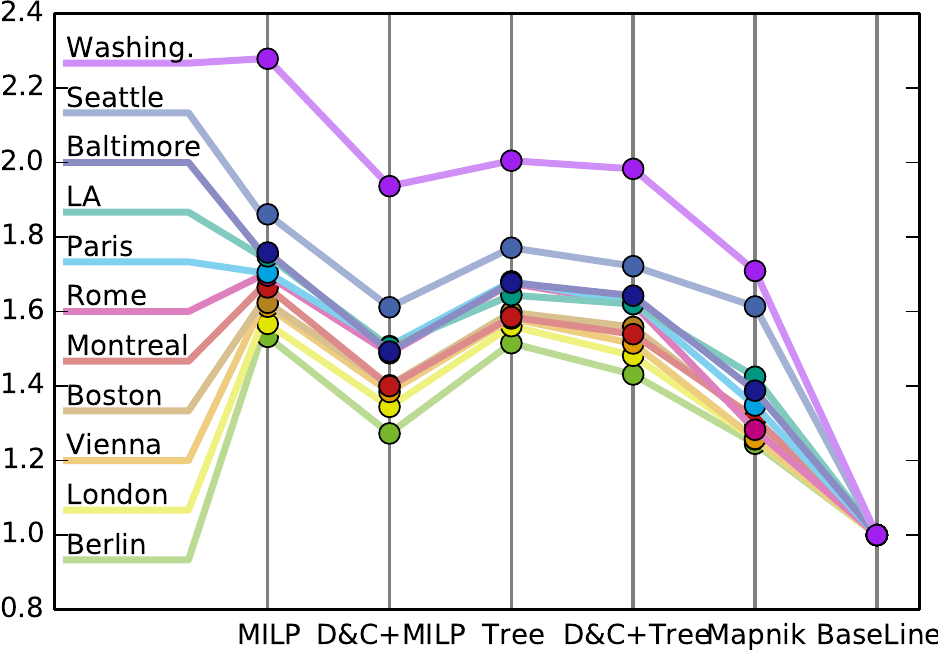}\label{apx:plot:rs-labels}}
\caption{\subref{apx:plot:rs-labels} Number of placed labels in thousands. \subref{apx:plot:rs-labels} Ratio of labeled road sections and placed labels: ${(\#\text{Labeled Road Sections})}/{(\#\text{Labels})}$.
}
\label{apx:fig:number-of-labels}
\end{figure}

\begin{figure}[t]
\centering
\includegraphics[scale=0.8]{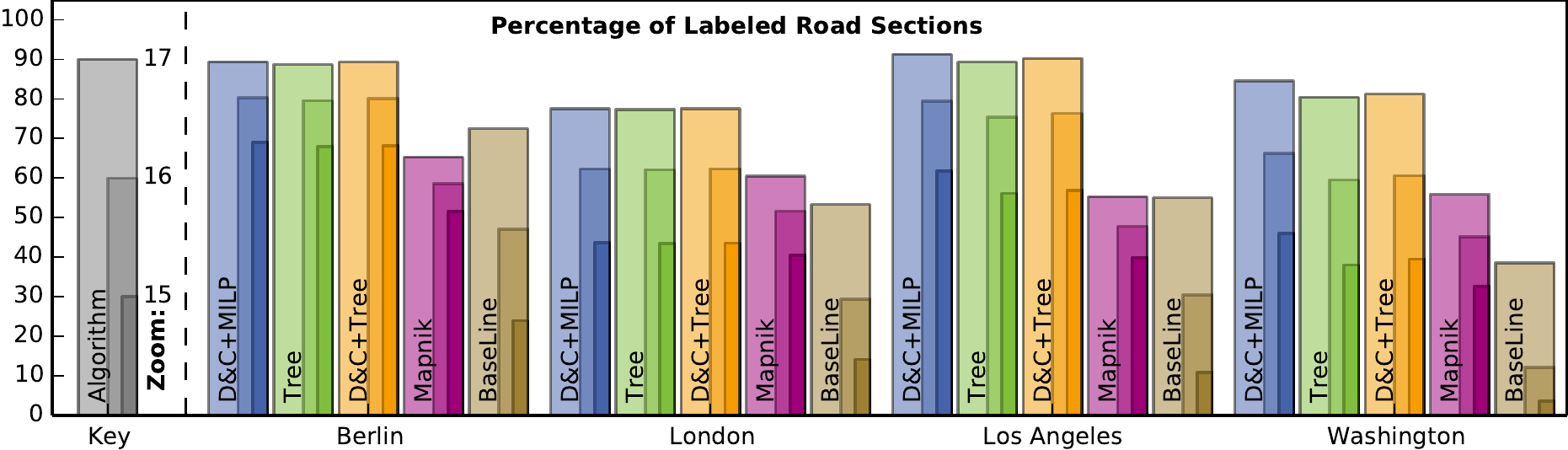}
\caption{Percentage of labeled road sections broken down in zoom levels and algorithms. The width of the bars (thin, medium, wide) corresponds to the zoom level  (15, 16, 17). 
}
\label{fig:quality:selection}
\vspace{-2ex}
\end{figure}

\textit{Phase 2, Quality.} First we analyze the average percentage of labeled
road sections over the three zoom levels; see
Fig.~\ref{plot:all:rs}. As an upper bound, \ILPAlgo, which provably 
solves \MaxTotalCovering optimally, yields results from $46.2\%$ (Rome) to
$80.3\%$ (Montreal). Considering zoom levels independently, we obtain
a minimum of $27.5\%$ (Rome, ZL 15) and a maximum of  $91.7\%$ (Montreal, ZL 17). We think that the wide span is attributed to the different
structures of road networks and road names, e.g., Rome has a lot of
short alleys and long road names. Hence, many road components are too
short or convoluted to contain a single label. Abbreviating road names
could help to overcome this problem.

The algorithm \Shredder+\TreeAlgo yields marginally better results
than \TreeAlgo, but only $1\%$ on average, see
Table~\ref{table:speedup}. Comparing \Shredder+\TreeAlgo with \ILPAlgo
we observe that \Shredder+\TreeAlgo yields near-optimal results with
respect to our road-section based model.  On average it reaches $97\%$
of the optimal solution; see Table~\ref{table:speedup}. While the
quality ratio is only $91\%$ for Washington, more than half of the
instances are labeled with a quality ratio of $\ge 99\%$.  For
European cities the percentage of road sections that belong to
components that are optimally solved by \TreeAlgo (long edges, paths,
and trees) is notably higher than those for North American cities; see
Fig.~\ref{fig:components}. Nonetheless, we obtain similar percentages
of labeled road sections for North American Cities. Hence, the
heuristic computing a spanning tree of non-tree components is both
fast and yields near-optimal results. The additional implementation
effort of \TreeAlgo is further justified by the observation that the
naive way to place labels only on single road sections lags far
behind; only $42\%$ on average, $58\%$ as maximum and $24\%$ as
minimum compared to the optimal solution.  Mapnik achieves on average
$75\%$ of the optimal solution and a maximum of $91\%$. For more than
the half of the instances Mapnik achieves at most $76\%$ of the
optimal solution. So in direct comparison, \Shredder+\TreeAlgo labels
$31\%$ more road sections than Mapnik on average. 
Moreover, \Shredder+\TreeAlgo has a better utilization of labels and achieves an average ratio of 1.61 labeled road sections per label, compared to Mapnik with a ratio of 1.37; see Fig.~\ref{apx:fig:number-of-labels}.


With increasing zoom level the number of
labeled road sections is increased, which is to be expected, since
more road sections become long-edges; see
Fig.~\ref{fig:quality:selection} for four cities (similar results apply for the others) and
Fig.~\ref{fig:components}. For each zoom level, we observe
similar results as described before: \TreeAlgo and \Shredder+\TreeAlgo
achieve near-optimal solutions and
Mapnik labels considerably fewer road sections. However, for smaller zoom
levels the gap between \ILPAlgo and Mapnik shrinks. 

From a visual perspective, labels lie on the skeleton of the road
network, which is achieved by design; see Fig.~\ref{fig:motivation}(c)
and the interactive maps. Instead of unnecessary repetition of labels,
labels are only placed if they actually convey additional
information. In particular, visual components are labeled, but not
single lanes that are indistinguishable due to the zoom level.

\section{Conclusion}
We introduced a generic framework for labeling road maps based on an
abstract road graph model that is combinatorial rather than geometric.
We showed in our experimental evaluation that our proposed heuristic
for decomposing the road graph into tree-shaped subgraphs and labeling
those trees provably optimally is both efficient and effective.  It
has running times in the range of seconds to one minute even for large
road networks such as London with more than 100,000 road sections and
achieves near-optimal quality ratios (on average 97\%) compared to
upper bounds computed by the exact method \ILPAlgo.  Our algorithm
clearly outperforms the labeling algorithm of the standard OSM
renderer Mapnik, with an average improvement in labeled road sections
of $31\%$.  Interestingly, \ILPAlgo is able to compute mathematically
optimal solutions within a few minutes for all our test instances,
even though it is slower by a factor of about 6 compared to the
tree-based algorithm. So for practical purposes there is a trade-off
between a final, but rather small improvement in quality at the cost
of a significant and by the very nature of \ILPAlgo unpredictable
increase in running time. We only implemented essential cartographic
criteria to evaluate the algorithmic core of our framework; further
criteria (e.g., abbreviated names) and alternative definitions of road sections can be easily incorporated.
The framework can further be pipelined with labeling
  algorithms for other map features, e.g., after placing labels for point features,
  one may block all parts of the road network covered by a point label and label the remaining road
  network such that no labels overlap. While this allows to label
  different types of features sequentially, constructing a labeling of
  all features in one single step remains an important open problem.

%
%

\medskip
\noindent\textbf{Acknowledgment.} We thank Andreas Gemsa for many interesting
and inspiring discussions, and his help on the implementation.

{
%

}
\end{document}